\documentclass[11pt,onecolumn,twoside]{IEEEtran} 
\usepackage{graphicx}
\usepackage{amsmath}
\usepackage{mathrsfs}
\usepackage{mathtools}
\usepackage{amssymb}
\usepackage{float}
\usepackage{url}
\usepackage{verbatim}
\usepackage{amsthm}
\usepackage{enumerate}
\usepackage{layout}
\usepackage{bm}

\newtheorem{theorem}{Theorem}
\newtheorem{lemma}{Lemma}

\newtheorem{definition}{Definition}
\theoremstyle{definition}

\newtheorem{remark}{Remark}

\newcommand{\Reals}{\ensuremath{\mathbb{R}}}
\newcommand{\expectation}{\ensuremath{\mathbb{E}}}
\newcommand{\Var}{\mathrm{Var}}

\newcommand{\set}{\ensuremath{\mathcal}}

\newcommand{\dif}{\mathrm{d}}
\newcommand{\A}{\ensuremath{\mathcal{A}}}

\begin{document}
\thispagestyle{empty}
\setcounter{page}{1}
\setlength{\baselineskip}{1.15\baselineskip}

\title{\huge{On Relations Between Tight Bounds for Symmetric $f$-Divergences and Binary Divergences}\\[0.2cm]}
\author{Tomohiro Nishiyama\\ Email: htam0ybboh@gmail.com}
\date{}
\maketitle
\thispagestyle{empty}

\begin{abstract}
Minimizing divergence measures under a constraint is an important problem. We derive a sufficient condition that binary divergence measures provide lower bounds for symmetric divergence measures under a given triangular discrimination or given means and variances. 
Assuming this sufficient condition, the former bounds are always tight, and the latter bounds are tight when two probability measures have the same variance. An application of these results for nonequilibrium physics is provided.

\end{abstract}
 
\section{Introduction}
Divergence measures such as the Kullback–Leibler divergence~\cite{kullback1951information}, Hellinger distance~\cite{hellinger1909neue} and triangular discrimination~\cite{le2012asymptotic} are widely used in information theory, statistics, machine learning, physics, and other theoretical and applied branches of mathematics.
They all belong to an important class of divergence measures, defined by means of convex functions $f$, and named $f$ -divergences~\cite{csiszar1967information,csiszar1967topological,csiszar1972class, sason2016f}.
Minimizing $f$-divergences under a constraint is an important problem. For arbitrary symmetric $f$-divergences, the binary $f$-divergences, which are $f$-divergences between two-element probability measures, provide tight lower bounds under a given total variation distance~\cite{gilardoni2006minimum,sason2014tight}. When means and variances of two probability measure are fixed, the binary $f$-divergences also provide tight lower bounds for the Kullback-Leibler divergence, Hellinger distance, and $\chi^2$-divergence~\cite{nishiyama2020relations, nishiyama2020tight, nishiyama2021tight}. 
In particular the tight bound for the $\chi^2$-divergence is known as the Hammersley-Chapman-Robbins bound~\cite{chapman1951minimum}.

The main purpose of this work is to derive a sufficient condition that the binary $f$-divergences provide lower bounds for symmetric $f$-divergences under a given triangular discrimination or given means and variances.
Assuming this sufficient condition, the former bounds are tight for arbitrary symmetric $f$-divergences although the latter bounds are not necessarily tight. We show the latter bounds reduce to tight bounds for arbitrary symmetric $f$-divergences when two probability measure have the same variance. These results allow us to derive novel inequalities for the triangular discrimination. Finally, we introduce an example of applying our results to derive an important relation in nonequilibrium physics, which was recently derived in \cite{Vo_2022}.

\section{Preliminaries}
This section provides definitions of divergence measures which are used in this note.
\begin{definition}
Let $f: (0,\infty)\rightarrow \Reals$ be a convex function with $f(1) = 0$. Let $P$ and $Q$ be probability measures defined on a measurable space $(\mathcal{X}, \mathscr{F})$, let $\mu$ be a dominating measure
of $P$ and $Q$ (i.e., $P, Q \ll \mu$), and let $p := \frac{\mathrm{d}P}{\mathrm{d}\mu}$
and $q := \frac{\mathrm{d}Q}{\mathrm{d}\mu}$ be the densities of $P$ and $Q$ with respect
to $\mu$. The {\em $f$-divergence} from $P$ to $Q$ is given by
\begin{align} \label{eq:fD}
D_f(P\|Q) := \int q \, f \Bigl(\frac{p}{q}\Bigr) \, \mathrm{d}\mu,
\end{align}
where
\begin{align}
& f(0) := \underset{t \to 0^+}{\lim} \, f(t), \quad  0 f\biggl(\frac{0}{0}\biggr) := 0, \nonumber \\[0.1cm] 
& 0 f\biggl(\frac{a}{0}\biggr) \nonumber
:= \lim_{t \to 0^+} \, t f\biggl(\frac{a}{t}\biggr)
= a \lim_{u \to \infty} \frac{f(u)}{u}, \quad a>0.
\end{align}
It should be noted that the right side of \eqref{eq:fD} does not depend on the dominating measure $\mu$.
\end{definition}
We define notable $f$-divergences: 
\begin{itemize}
\item {\em Triangular discrimination} (a.k.a. Vincze-Le Cam distance): $f(t) = \frac{(1-t)^2}{2(1+t)}$,
\begin{align}
\Delta(P,Q):=\frac12\int \frac{(p-q)^2}{(p+q)}\dif\mu.
\end{align}
\item {\em Kullback-Leibler divergence} (a.k.a. relative entropy): $f(t) = t\log t$,
\begin{align}
D(P\|Q):=\int p\log\frac{p}{q} \dif\mu.
\end{align}

\item {\em Squared  Hellinger distance}: $f(t) = \frac12(\sqrt{t}-1)^2$,
\begin{align}
H^2(P,Q):=\frac12\int (\sqrt{p}-\sqrt{q})^2\dif\mu.
\end{align}

\item {\em Jensen-Shannon divergence} (a.k.a. capacitory discrimination): $f(t) :=\frac12 t\log t-\frac{(1+t)}2\log\frac{1+t}2$,
\begin{align}
\mathrm{JS}(P,Q):=\frac12\Bigl(D\Bigl(P\|\frac{P+Q}2\Bigr)+D\Bigl(Q\|\frac{P+Q}2\Bigr)\Bigr).
\end{align}
\end{itemize}
In this note, logarithms have an arbitrary common base except Subsection~\ref{subsec_thermo}.

\begin{definition}
Let $P$ and $Q$ be probability measures defined on the measurable space $(\Reals,\mathscr{B})$, where $\Reals$ is the real
line and $\mathscr{B}$ is the Borel $\sigma$–algebra of subsets of $\Reals$.
Let $\set{P}[m_P, \sigma_P; m_Q, \sigma_Q]$ be a set of pairs of probability measures $(P,Q)$ under given means and variances, i.e.,
\begin{align}
\label{constraints}
& \expectation[X] =: m_P, \; \expectation[Y] =: m_Q,
\quad \Var[X] =: \sigma_P^2, \;  \Var[Y] =: \sigma_Q^2,
\end{align}
where $X\sim P$ and $Y\sim Q$.
\end{definition}

\begin{definition} 
\label{def:binary f div}
For $t\in[0,1]$, let $R_t:=\Bigl(\frac{1-t}2,\frac{1+t}2\Bigr)$ and $R_t^\dagger:=\Bigl(\frac{1+t}2,\frac{1-t}2\Bigr)$ be two-element probability measures defined on the same set.
Letting  $g:[0,1]\rightarrow [0,\infty)$, we define the binary $f$-divergence between $R_t$ and $R^\dagger_t$ by
\begin{align}
\label{def_binary}
g(t):=D_f(R_t\|R_t^\dagger)=\frac{(1-t)}{2} f\Bigl(\frac{1+t}{1-t}\Bigr)+\frac{(1+t)}{2} f\Bigl(\frac{1-t}{1+t}\Bigr).
\end{align}

\end{definition}

\section{Main results}
\label{sec_main}
Throughout this note, we suppose that $f$ is twice differentiable and strictly convex.
\begin{theorem}
Let $P$ and $Q$ be probability measures defined on a measurable space $(\mathcal{X}, \mathscr{F})$.
If $\frac1t g'(t)$ is non-decreasing, for every $d\in[0,1]$, 
\label{th_1}
\begin{align}
\label{th1_1}
\min_{(P,Q): \Delta(P,Q)=d}\frac12\Bigl(D_f(P\|Q)+ D_f(Q\|P)\Bigr)= g(\sqrt{d}).
\end{align}
The lower bound is attained by $(P,Q)=(R_{\sqrt{d}},R^\dagger_{\sqrt{d}})$.
\end{theorem}
Here, $'$ denotes the derivative with respect to $t$.

\begin{theorem} \label{th_2}
If $\frac1t g'(t)$ is non-decreasing, 
\begin{align}
\label{th2_1}
\inf_{(P,Q)\in\set{P}[m_P, \sigma_P; m_Q, \sigma_Q]}\frac12\Bigl(D_f(P\|Q)+ D_f(Q\|P)\Bigr)\geq g(s),
\end{align}
where 
\begin{align}
\label{th2_2}
s&:=\frac{|a|}{\sqrt{2(\sigma_P^2+\sigma_Q^2)+a^2}}, \\
a&:=m_P-m_Q.
\end{align}
If $\sigma_P=\sigma_Q=:\sigma$, the inequality \eqref{th2_1} reduces to 
\begin{align}
\label{th_2_3}
\min_{(P,Q)\in\set{P}[m_P, \sigma; m_Q, \sigma]}\frac12\Bigl(D_f(P\|Q)+ D_f(Q\|P)\Bigr) = g(r), 
\end{align}
where 
\begin{align}
\label{th_2_4}
r:=\frac{|a|}{\sqrt{4\sigma^2+a^2}}.
\end{align}
The lower bound is attained by $(P,Q)=(R_r,R^\dagger_r)$.
\end{theorem}

\subsection{Proofs}
We prove the following lemmas before proving theorems.
\label{section: proofs}
\begin{lemma}
\label{lem_1}
The binary $f$-divergence $g$ is strictly increasing. 
Let $G: [0,\infty)\rightarrow [0,1]$ be an inverse function of $g$. If $\frac1t g'(t)$ is non-decreasing, the square of $G$ is concave.
\end{lemma}
\begin{proof}
From \eqref{def_binary}, we have
\begin{align}
\label{pr1_5}
g'(t)&=\frac12\Bigl(f(v(t))-f(v(-t))\Bigr)+\frac{1}{1-t}f'(v(-t))-\frac{1}{1+t}f'(v(t)), 
\end{align}
where $v(t):=\frac{1-t}{1+t}=\frac{2}{1+t}-1$.
From the strictly convexity of $f$, we have $f(v(t))-f(v(-t))> f'(v(-t))(v(t)-v(-t))$ and $f'(v(-t))-f'(v(t))>0$.
By combining these relations with \eqref{pr1_5}, we have
\begin{align}
\label{pr1_7}
g'(t)> \frac{1}{(1+t)}\Bigl(f'(v(-t))-f'(v(t))\Bigr) > 0.
\end{align}
Hence, the function $g$ is strictly increasing and the inverse function $G$ exists. Differentiating $G(T)^2$ with respect to $T:=g(t)$, we have
\begin{align}
\label{pr1_1}
\frac{\dif^2 (G(T)^2)}{\dif T^2}&=2\Bigl(\frac{\dif G(T)}{\dif T}\Bigr)^2+2G(T)\frac{\dif^2 G(T)}{\dif T^2}.
\end{align}
Furthermore,
\begin{align}
\label{pr1_2}
\frac{\dif G(T)}{\dif T}&=\frac{\dif t}{\dif T}=\frac{1}{g'(t)}, \\
\label{pr1_3}
\frac{\dif^2 G(T)}{\dif T^2}&=-\frac{g''(t)}{g'(t)^3}.
\end{align}
Substituting \eqref{pr1_2} and \eqref{pr1_3} into \eqref{pr1_1}, we have
\begin{align}
\label{pr1_4}
\frac{\dif^2 (G(T)^2)}{\dif T^2}=\frac{2}{g'(t)^3}\Bigl(g'(t)-tg''(t)\Bigr),
\end{align}
where we use $t=G(T)$.
From the assumption, we have $tg''(t)-g'(t)\geq 0$ for all $t\in(0,1)$.
By combining this inequality with \eqref{pr1_7} and \eqref{pr1_4}, we have $\frac{\dif^2 (G(T)^2)}{\dif T^2}\leq 0$.
Hence, $G(T)^2$ is concave.

\end{proof}

\begin{lemma}{(Sedrakyan's inequality)}
\label{lem_2}
Let $u,v: \mathcal{X}\rightarrow \Reals$ and $v(x)>0$ for al $x\in\mathcal{X}$.
Then, 
\begin{align}
\label{lem2_1}
\int \frac{u^2}{v}\dif\mu  \geq \frac{\Bigl(\int u\dif\mu\Bigr)^2}{\int v\dif\mu}.
\end{align}
The equality holds if and only if there exists a constant $c\in\Reals$ such that $u=cv$ for all $x\in\mathcal{X}$.
\end{lemma}
\begin{proof}
Applying the Cauchy-Schwarz inequality for $\frac{u}{\sqrt{v}}$ and $\sqrt{v}$, it follows that
\begin{align}
\int \frac{u^2}{v}\dif\mu \int v \dif\mu\geq \Bigl(\int u \dif\mu\Bigr)^2.
\end{align}
Hence, we obtain \eqref{lem2_1} and the equality condition.
\end{proof}

\begin{lemma}
\label{lem_3}
For $x\in\Reals$ and $r$ defined by \eqref{th_2_4}, let $(R_r, R_r^\dagger)$ be a pair of two-element probability measures defined on $\{-x\frac{a}{|a|}+\frac{m_P+m_Q}2, x\frac{a}{|a|}+\frac{m_P+m_Q}2\}$, where $x:=\frac{\sqrt{4\sigma^2+a^2}}2$ and $\frac{a}{|a|}:=1$ when $a=0$.
Then, $(R_r, R_r^\dagger)\in\set{P}[m_P, \sigma; m_Q, \sigma]$.
\end{lemma}
\begin{proof}
It can be easily verified that $(R_r, R_r^\dagger)$ defined on $\{-x\frac{a}{|a|}, x\frac{a}{|a|}\}$ belongs to $\set{P}[\frac{a}2, \sigma; -\frac{a}2, \sigma]$. Since $\expectation[X+c]=\expectation[X]+c$ and $\Var[X+c]=\Var[X]$ for $c=\frac{m_P+m_Q}2$, the result follows. 
\end{proof}

\subsubsection{Proof of Theorem~\ref{th_1}}  \mbox{}\\
\label{proof_th_1}
\begin{proof}
For $p,q\in[0,\infty)$, substituting $t=\frac{|p-q|}{p+q}\leq 1$ into $G(g(t))^2=t^2$ and using \eqref{def_binary}, we have
\begin{align}
\label{pr2_1}
G\Bigl(g\Bigl(\frac{p-q}{p+q}\Bigr)\Bigr)^2=G\Bigl(\frac{q}{p+q}f\Bigl(\frac{p}{q}\Bigr)+\frac{p}{p+q}f\Bigl(\frac{q}{p}\Bigr)\Bigr)^2=\frac{(p-q)^2}{(p+q)^2}.
\end{align}
Multiplying \eqref{pr2_1} by $\frac{p+q}2$ and integrating it, we have 
\begin{align}
\label{pr2_2}
\int \frac{(p+q)}2G\Bigl(\frac{q}{p+q}f\Bigl(\frac{p}{q}\Bigr)+\frac{p}{p+q}f\Bigl(\frac{q}{p}\Bigr)\Bigr)^2\dif \mu=\int \frac{(p-q)^2}{2(p+q)}\dif \mu=\Delta(P,Q)=d.
\end{align}
From Lemma~\ref{lem_1}, by applying the Jensen's inequality for $G(\cdot)^2$, we have
\begin{align}
\label{pr2_3}
d=\int \frac{(p+q)}2G\Bigl(\frac{q}{p+q}f\Bigl(\frac{p}{q}\Bigr)+\frac{p}{p+q}f\Bigl(\frac{q}{p}\Bigr)\Bigr)^2\dif \mu\leq G\Bigl(\frac12\Bigl(D_f(P\|Q)+D_f(Q\|P)\Bigr)\Bigr)^2,
\end{align}
where we use the definition of $f$-divergence and $\frac12\int (p+q) \dif \mu=1$. 
Since $g$ is monotonically increasing from Lemma~\ref{lem_1}, taking the square root of both sides and applying the function $g$, we have
\begin{align}
\label{pr2_4}
\frac12\Bigl(D_f(P\|Q)+D_f(Q\|P)\Bigr) \geq g(\sqrt{d}).
\end{align}
Considering $(P,Q)=(R_{\sqrt{d}},R^\dagger_{\sqrt{d}})$, we have $\Delta(R_{\sqrt{d}},R^\dagger_{\sqrt{d}})=d$ and $D_f(R_{\sqrt{d}}\|R^\dagger_{\sqrt{d}})=D_f(R^\dagger_{\sqrt{d}}\|R_{\sqrt{d}})=g(\sqrt{d})$ from Definition~\ref{def:binary f div}.
Hence, the equality holds in \eqref{pr2_4}.
\end{proof}
\remark Considering $G(\cdot)$ instead of $G(\cdot)^2$ in this proof, we obtain the tight lower bounds for arbitrary symmetric $f$-divergences under a given total variation distance, which was shown in~\cite{sason2014tight}.
\subsubsection{Proof of Theorem~\ref{th_2}} \mbox{}\\
\begin{proof}
Under transformation $x\rightarrow x -\frac{m_P+m_Q}2$, $f$-divergences are invariant. Without any loss of generality, one can assume $(P,Q)\in \set{P}[\frac{a}2, \sigma_P; -\frac{a}2, \sigma_Q]$. By applying Lemma~\ref{lem_2} for \eqref{pr2_2} and combining the result with \eqref{pr2_3}, we have 
\begin{align}
G\Bigl(\frac12\Bigl(D_f(P\|Q)+D_f(Q\|P)\Bigr)\Bigr)^2\geq \int \frac{(p-q)^2}{2(p+q)}\dif \mu\geq \int_{x\neq 0} \frac{(x p-x q)^2}{2(x^2 p+x^2 q)}\dif \mu \nonumber \\
\geq  \frac{\Bigr( \int_{x\neq 0} (x p-xq)\dif\mu\Bigr)^2}{2\int_{x\neq 0} (x^2 p+x^2 q)\dif\mu}=\frac{a^2}{2(\sigma_P^2+\sigma_Q^2)+a^2}=s^2.
\end{align}
In the similar way to the proof of Theorem~\ref{th_1}, we obtain \eqref{th2_1}.
When $\sigma_P=\sigma_Q$, considering $(P,Q)=(R_{r},R^\dagger_{r})$, we have $D_f(R_{r}\|R^\dagger_{r})=D_f(R^\dagger_{r}\|R_{r})=g(r)$ from Definition~\ref{def:binary f div} and $(R_r,R^\dagger_{r})\in\set{P}[m_P, \sigma; m_Q, \sigma]$ from Lemma~\ref{lem_3}. 
Hence, we obtain \eqref{th_2_3}.

\end{proof}
\section{Applications of main results}
\label{sec_appli}
\subsection{Inequalities for the triangular discrimination}
We show inequalities between the triangular discrimination and other symmetric $f$-divergences from Theorem~\ref{th_1}.
\begin{itemize}
\item Squared Hellinger distance:\\
Since the binary squared Hellinger distance is defined by $g(t)=1-\sqrt{1-t^2}$ and $\frac1t g'(t)=\frac{1}{\sqrt{1-t^2}}$ is non-decreasing, we obtain 
\begin{align}
\label{ex_1}
H^2(P,Q)\geq 1-\sqrt{1-\Delta(P,Q)}=\frac{\Delta(P,Q)}{1+\sqrt{1-\Delta(P,Q)}}.
\end{align}
Since $\Delta(P,Q)\in[0,1]$, this inequality is tighter than $H^2(P,Q)\geq \frac12 \Delta(P,Q)$ in~\cite{le2012asymptotic}. From \eqref{ex_1}, we also have
\begin{align}
\Delta(P,Q)+Z(P,Q)^2\leq 1,
\end{align}
where $Z(P,Q):=\int \sqrt{pq} \dif \mu$ denotes the Bhattacharyya coefficient~\cite{kailath1967divergence}.

\item Jensen-Shannon divergence:\\
Since the binary Jensen-Shannon divergence is defined by $g(t)=\frac{(1+t)}2\log(1+t)+\frac{(1-t)}2\log(1-t)$ and  $\frac1t g'(t)=\frac1{2t} \log\frac{1+t}{1-t}$ is non-decreasing, we obtain
\begin{align}
\label{ex_2}
\mathrm{JS}(P,Q)\geq g(\sqrt{\Delta(P,Q)})=\log 2-H_b(R_{\sqrt{\Delta(P,Q)}}),
\end{align}
where $H_b(R_t):=-\frac{(1-t)}2\log \frac{1-t}2 - \frac{(1+t)}2\log\frac{1+t}2$ denotes the binary entropy. The right-hand side of \eqref{ex_2} was calculated numerically by solving an optimization problem~\cite{guntuboyina2013sharp}. From $g'(\sqrt{t})=\frac{1}{4\sqrt{t}}\log\frac{(1+\sqrt{t})}{(1-\sqrt{t})}$ and $\log\frac{(1+t)}{(1-t)} \geq 2(\log e) t$, we have $g'(\sqrt{t})\geq \frac{\log e}2$. By combining this relation with $g(0)=0$, we have $g(\sqrt{t})\geq \frac{\log e}2 t$. Therefore, the inequality \eqref{ex_2} is tighter than $\mathrm{JS}(P,Q)\geq \frac{\log e}2\Delta(P,Q)$ in~\cite{topsoe2000some}. 
\end{itemize}

Next, we show a tight bound for the triangular discrimination under given means and variances. Since the binary triangular discrimination is given by $g(t)=t^2$, from Theorem~\ref{th_2}, it follows that
\begin{align}
\Delta(P,Q)\geq s^2,
\end{align}
where $s$ is defined by \eqref{th2_2}. This inequality was shown in~\cite{falasco2022beyond}.
Let $P=(p,1-p)$ and $Q=(q,1-q)$ be two-element probability measures defined on $\{x_1+\frac{m_P+m_Q}2, x_2+\frac{m_P+m_Q}2\}$, and $(P,Q)\in \set{P}[m_P, \sigma_P; m_Q, \sigma_Q]$. The existence of $(P,Q)$ was shown in~\cite{nishiyama2020relations}. We show that the binary triangular discrimination $\Delta(P,Q)$ attains the lower bound. 
From the assumption, we have
\begin{align}
p(x_1-x_2)+x_2=\frac{a}2, \; q(x_1-x_2)+x_2=-\frac{a}2.
\end{align}
Taking the sum of these relations, we have
\begin{align}
(p+q)(x_1-x_2)+2x_2=0,  \;  (1-p+1-q)(x_2-x_1)+2x_1=0.
\end{align}
Therefore, there exists $c=(p-q)\frac{(x_2-x_1)}{2x_1x_2}$ such that
\begin{align}
p-q=c(p+q)x_1, \; 1-p-(1-q)=c(1-p+1-q)x_2.
\end{align}
Since this relation is the equality condition in Lemma~\ref{lem_2}, we have 
\begin{align}
\Delta(P,Q)=\frac{\{(p-q)x_1\}^2}{2(p+q)x_1^2}+\frac{\{(1-p-(1-q))x_2\}^2}{2(1-p+1-q)x_2^2}=\frac{a^2}{2(\sigma_P^2+\sigma_Q^2)+a^2}=s^2.
\end{align}

\subsection{Application for nonequilibrium physics}
\label{subsec_thermo}
As a further application of Theorem~\ref{th_1}, we show the relation among the total entropy production, pseudo-entropy production~\cite{shiraishi2021optimal} and dynamical activity in nonequilibrium physics. This is the key relaiton to derive the unified thermodynamic-kinetic uncertainty relation~\cite{Vo_2022}. We consider a discrete-state system described by a  continuous-time Markov jump process. The time evolution of the system is given by
\begin{align}
\label{tkur_1_1}
\frac{\mathrm{d}p_n(t)}{\mathrm{d}t} = \sum_{m} p_m(t)R_{nm},
\end{align}
where $p_n(t)$ denotes the probability in state $n$ at time $t$, and $R_{nm}$ denotes the transition rate from state $m$ to state $n$. We suppose that the transition rate is time-independent. The transition rate satisfies $\sum_{n} R_{nm}=0$ and $R_{nm}\geq 0$ for $n\neq m$. 
Assuming the local detailed balance condition, the total entropy production $\Sigma_\tau$ and the pseudo-entropy production  $\Sigma_\tau^{\mathrm{ps}}$ are defined by
\begin{align}
\label{tkur_1_2}
\Sigma_\tau &:= \int_0^\tau \mathrm{d}t  \sum_{n<m} (K_{nm}(t)-K_{mn}(t))\ln\frac{K_{nm}(t)}{K_{mn}(t)}, \\
\Sigma_\tau^{\mathrm{ps}} &:= 2\int_0^\tau \dif t\sum_{n<m} \frac{(K_{nm}(t)-K_{mn}(t))^2}{a_{nm}(t)},
\end{align}
where $K_{nm}(t):=p_m(t)R_{nm}$ and $a_{nm}(v,t):=K_{nm}(t)+K_{mn}(t)$ denotes the jump frequency associated with the transition between $m$ and $n$. 
Let $\mathcal{S}$ be a set of states, and let $P$ and $P^\dagger$ be probability measures defined on the set $\{(n, m, t) | n,m \in \mathcal{S}, n\neq m, t \in [0, \tau]\}$ as follows.
\begin{align}
\label{tkur_1_3}
&P(n,m,t):=\frac{1}{\A_\tau} K_{nm}(t), \\
\label{tkur_1_4}
&P^\dagger(n,m,t):=P(m,n,t), \\
&\int_0^\tau\mathrm{d}t\sum_{n\neq m} P(n,m,t)=1,
\end{align}
where 
\begin{align}
\label{tkur_1_5}
\A_\tau := \int_0^\tau \mathrm{d}t  \sum_{n<m} a_{nm}(t),
\end{align}
denotes the dynamical activity.
This definition connects the total entropy production and the pseudo-entropy production with the Kullback-Leibler divergence and the triangular discrimination, respectively.
\begin{align}
\label{tkur_1_6}
\Sigma_\tau &= \A_\tau D(P\|P^\dagger), \\
\Sigma_\tau^{\mathrm{ps}} &= 2\A_\tau \Delta(P,P^\dagger).
\end{align}
The binary Kullback-Leibler divergence is given by $g(t)=t\ln \frac{1+t}{1-t}$ and $\frac1t g'(t)=\frac1t \ln \frac{1+t}{1-t}+\frac{2}{1-t^2}$ is non-decreasing. From Theorem~\ref{th_1} and $D(P\|P^\dagger)=D(P^\dagger \|P)$, we obtain
\begin{align}
\Sigma_\tau\geq \A_\tau g\Bigl(\sqrt{\frac{\Sigma_\tau^{\mathrm{ps}}}{2\A_\tau}}\Bigr).
\end{align}
This relation is equivalent to the inequality in~\cite{Vo_2022} (see~\cite{nishiyama2022thermo}).
\section{Conclusion}
We have derived the sufficient condition that the binary $f$-divergences provide lower bounds for symmetric $f$-divergences under a given triangular discrimination or given means and variances. Assuming this sufficient condition, the former bounds are tight, and the latter bounds are tight when two probability measures have the same variance. From these results, we have derived novel inequalities for the triangular discrimination and we have given an another derivation of the relation among the total entropy production, pseudo-entropy production and dynamical activity in nonequilibrium physics.

Future work is to extend our results to asymmetric $f$-divergences and to study conditions that the binary $f$-divergences attain lower bounds when two probability measures have different variances.

\bibliography{reference_f_div} 

\begin{thebibliography}{10}

\bibitem{chapman1951minimum}
D.~G. Chapman, H.~Robbins, et~al.
\newblock Minimum variance estimation without regularity assumptions.
\newblock {\em The Annals of Mathematical Statistics}, 22(4):581--586, 1951.

\bibitem{csiszar1967information}
I.~Csisz{\'a}r.
\newblock Information-type measures of difference of probability distributions
  and indirect observation.
\newblock {\em studia scientiarum Mathematicarum Hungarica}, 2:229--318, 1967.

\bibitem{csiszar1967topological}
I.~Csisz{\'a}r.
\newblock On topological properties of f-divergences.
\newblock {\em Studia Math. Hungar.}, 2:329--339, 1967.

\bibitem{csiszar1972class}
I.~Csisz{\'a}r.
\newblock A class of measures of informativity of observation channels.
\newblock {\em Periodica Mathematica Hungarica}, 2(1-4):191--213, 1972.

\bibitem{falasco2022beyond}
G.~Falasco, M.~Esposito, and J.-C. Delvenne.
\newblock Beyond thermodynamic uncertainty relations: nonlinear response,
  error-dissipation trade-offs, and speed limits.
\newblock {\em Journal of Physics A: Mathematical and Theoretical},
  55(12):124002, 2022.

\bibitem{gilardoni2006minimum}
G.~L. Gilardoni.
\newblock On the minimum f-divergence for given total variation.
\newblock {\em Comptes Rendus Mathematique}, 343(11-12):763--766, 2006.

\bibitem{guntuboyina2013sharp}
A.~Guntuboyina, S.~Saha, and G.~Schiebinger.
\newblock Sharp inequalities for $ f $-divergences.
\newblock {\em IEEE transactions on information theory}, 60(1):104--121, 2013.

\bibitem{hellinger1909neue}
E.~Hellinger.
\newblock Neue begr{\"u}ndung der theorie quadratischer formen von
  unendlichvielen ver{\"a}nderlichen.
\newblock {\em Journal f{\"u}r die reine und angewandte Mathematik (Crelles
  Journal)}, 1909(136):210--271, 1909.

\bibitem{kailath1967divergence}
T.~Kailath.
\newblock The divergence and bhattacharyya distance measures in signal
  selection.
\newblock {\em IEEE transactions on communication technology}, 15(1):52--60,
  1967.

\bibitem{kullback1951information}
S.~Kullback and R.~A. Leibler.
\newblock On information and sufficiency.
\newblock {\em The annals of mathematical statistics}, 22(1):79--86, 1951.

\bibitem{le2012asymptotic}
L.~Le~Cam.
\newblock {\em Asymptotic methods in statistical decision theory}.
\newblock Springer Science \& Business Media, 2012.

\bibitem{nishiyama2020tight}
T.~Nishiyama.
\newblock A tight lower bound for the hellinger distance with given means and
  variances.
\newblock {\em arXiv preprint arXiv:2010.13548}, 2020.

\bibitem{nishiyama2021tight}
T.~Nishiyama.
\newblock Tight lower bounds for $\alpha $-divergences under moment constraints
  and relations between different $\alpha$.
\newblock {\em arXiv preprint arXiv:2105.12972}, 2021.

\bibitem{nishiyama2022thermo}
T.~Nishiyama.
\newblock Thermodynamic-kinetic uncertainty relation: properties and an
  information-theoretic interpretation.
\newblock {\em arXiv preprint arXiv:2207.08496}, 2022.

\bibitem{nishiyama2020relations}
T.~Nishiyama and I.~Sason.
\newblock On relations between the relative entropy and $\chi$2-divergence,
  generalizations and applications.
\newblock {\em Entropy}, 22(5):563, 2020.

\bibitem{sason2014tight}
I.~Sason.
\newblock Tight bounds for symmetric divergence measures and a refined bound
  for lossless source coding.
\newblock {\em IEEE Transactions on Information Theory}, 61(2):701--707, 2014.

\bibitem{sason2016f}
I.~Sason and S.~Verd{\'u}.
\newblock $ f $-divergence inequalities.
\newblock {\em IEEE Transactions on Information Theory}, 62(11):5973--6006,
  2016.

\bibitem{shiraishi2021optimal}
N.~Shiraishi.
\newblock Optimal thermodynamic uncertainty relation in markov jump processes.
\newblock {\em Journal of Statistical Physics}, 185(3):1--15, 2021.

\bibitem{topsoe2000some}
F.~Topsoe.
\newblock Some inequalities for information divergence and related measures of
  discrimination.
\newblock {\em IEEE Transactions on information theory}, 46(4):1602--1609,
  2000.

\bibitem{Vo_2022}
V.~T. Vo, T.~V. Vu, and Y.~Hasegawa.
\newblock Unified thermodynamic{\textendash}kinetic uncertainty relation.
\newblock {\em Journal of Physics A: Mathematical and Theoretical},
  55(40):405004, 2022.

\end{thebibliography}
\bibliographystyle{myplain} 


\end{document}